\newif\ifreport\reporttrue
\documentclass[conference, 10pt]{IEEEtran}

\usepackage{setspace}
\usepackage{cite}
\usepackage{graphicx}
\usepackage{caption}
\usepackage{subcaption}
\usepackage[cmex10]{amsmath}
\usepackage{mathtools}
\usepackage{amsthm}
\usepackage{amsfonts}
\usepackage{amssymb}
\usepackage{algorithm}
\usepackage{algorithmic}
\usepackage{bm,xstring}
\usepackage{fixltx2e}%fixes latex float algorithm in 2 columns
\usepackage{tikz,hyperref}
\usetikzlibrary{decorations.pathreplacing}
\newcommand{\age}{\Delta}
\usepackage{color}

\newcommand{\ignore}[1]{}
\usepackage[singlelinecheck=false]{caption}

\newtheorem{lemma}{Lemma}

\newtheorem{theorem}{Theorem}

\theoremstyle{definition}

\newtheorem{definition}{Definition}

\begin{document}

\title{Timely Updates: An Information Theoretical Approach} 
\title{Timely Information in Samples}
\title{Timely Information through Queues}
\title{A perspective on Information Aging through Queues}
\title{Information Aging through Queues: A Mutual Information Perspective}
%\title{Dynamic Sampling of Brownian Motion for }

%\title{Optimal Sampling of Brownian Motion in a Random Channel}

\IEEEoverridecommandlockouts
%\title{Update or Wait: How to Keep Your Data Fresh}
\author{
Yin Sun and Benjamin Cyr\\ %, Yury Polyanskiy, and Elif Uysal-Biyikoglu\\
Dept. of ECE,  Auburn University, Auburn, AL\\
%%Dept. of EECS, Massachusetts Institute of Technology, Cambridge, MA\\
%%Dept. of EEE, Middle East Technical University, Ankara, Turkey\\
%
%%\thanks{This work was supported in part by DTRA grant HDTRA1-14-1-0058, NSF grants CNS-1446582, CNS-1409336, CNS-1518829,  CNS-1514260, CNS-1422988, and CNS-1054738, ARO grant W911NF-14-1-0368, and ONR grant N00014-15-1-2166. E. Uysal-Biyikoglu was supported in part by TUBITAK and in part by a Science Academy BAGEP award.
\thanks{This work was supported in part by NSF grant CCF-1813050 and ONR grant N00014-17-1-2417.}
}
\maketitle
%\thispagestyle{plain}
%\pagestyle{plain}

% !TEX root = ./sampling_BM.tex
\begin{abstract}
In this paper, we propose a new measure for the  freshness of information, which uses the mutual information between the real-time source value and the delivered samples at the receiver to quantify the freshness of the information contained in the delivered samples. Hence, the ``aging" of the received  information can be  interpreted as a procedure that the above mutual information reduces as the age grows.
In addition, we  consider a sampling problem, where samples of a Markov source are taken and sent through a queue to  the receiver. In order to optimize the freshness of information, we study the optimal sampling policy that maximizes the time-average expected mutual information.  We prove that the optimal sampling policy is a threshold policy and find the optimal threshold exactly. Specifically, a new sample is taken once a conditional mutual information term reduces to a threshold, and the threshold is  equal to the optimum value of the time-average expected mutual information that is being maximized. Numerical results are provided to compare different sampling policies.
\end{abstract}

\ignore{
In this paper, we study a sampling problem, where samples of a Markov source are taken and sent through a queue to a destination. We use the mutual information between the real-time source value and the delivered samples to characterize the freshness of the information at the destination. Hence, ?aging of the delivered information" can be considered as a process that the above mutual information reduces as the age grows. We obtain the optimal-sampling policy that maximizes the time-average mutual information subject to a sampling rate constraint. The optimal sampling policy is proven to be a threshold policy, where the optimal threshold is obtained exactly. Numerical results are provided to compare different sampling policies.}

\section{Introduction}

Information usually has the greatest value when it is fresh \cite{Shapiro1999}. For example, real-time knowledge about the location, orientation, and speed of motor vehicles is imperative in autonomous driving, and the access to timely updates about the stock price and interest-rate movements is essential for developing trading strategies on the stock market. In \cite{Song1990,KaulYatesGruteser-Infocom2012}, the concept of \emph{Age of Information} was introduced  to measure the freshness of  information that a receiver has about the status of a remote source. Consider a sequence of source samples that are sent through a queue to a receiver, as illustrated  in Fig. \ref{fig_model}. Each sample is stamped with its generation time. 
Let $U_n$ be the time stamp of the newest sample that has been delivered to the receiver by time instant $n$. {The age of information, as a function of  $n$, is defined as} 
$\Delta_n = n - U_n$, which is the time elapsed since the newest sample was generated. 
Hence, a small  age $\Delta_n$ indicates that there exists a fresh sample of the source status  at the receiver. 

%By maintaining a small age $\Delta_n$ over time, one can ensure that a fresh sample of the status source is available at the receiver. 

In practice, the status of different sources may vary over time with different speeds. For example, the location of a car can change much faster than the temperature of its engine. While the age of information $\Delta_n$ represents the time difference between the samples available at the transmitter and receiver, it is independent of the changing speed of the source. Hence,  the age $\Delta_n$ is not  an appropriate measure for comparing the freshness of information about different sources. 

\begin{figure}%[!t]
\centering
\includegraphics[width=0.42\textwidth]{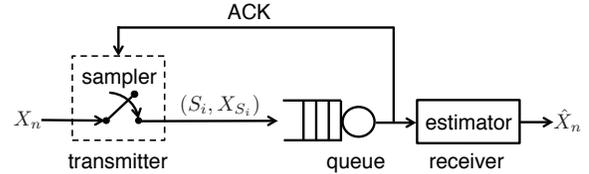}   
\caption{System model.}
\label{fig_model}
%\vspace{-6mm}
\end{figure}    
In recent years, several examples and approaches for evaluating the freshness of information about time-correlated sources have been discussed in, e.g., \cite{Yates2012,SunISIT2017,GAO201857,AgeOfInfo2016,Kosta2017,Bedewy2016,Bedewy2017,BedewyJournal2017,BedewyJournal2017_2,multiflow18}. In \cite{Yates2012,SunISIT2017,GAO201857} and the references therein, 
the received samples  are used to estimate the source value in  real-time, where the estimation error is used to measure  the freshness of information available at the receiver. 
 In \cite{AgeOfInfo2016}, an age penalty function $p(\Delta)$ was employed to describe the level of dissatisfaction for having aged samples at the receiver, where $p$ is an \emph{arbitrary} \emph{non-negative} and \emph{non-decreasing function} of the age $\Delta$ that can be specified based on the application; in addition, an optimal sampling strategy was developed to minimize the time-average expected age penalty function. In \cite{Kosta2017}, the authors considered the relationship between the auto-correlation function $r(\Delta_n) = \mathbb{E}[X_n^* X_{n-\Delta_n}]$ (where $X_n$ denotes the source status at time instant $n$) and the age penalty function in \cite{AgeOfInfo2016}, and provided analytical expressions for the long-run time average of a few auto-correlation functions. In \cite{Bedewy2016,Bedewy2017,BedewyJournal2017,BedewyJournal2017_2,multiflow18}, several  scheduling policies were developed to minimize an \emph{arbitrary non-decreasing functional} $f(\{\Delta_n: n\geq 0\})$ of the age process $\{\Delta_n: n\geq0\}$ in several network settings.  The age penalty models in  \cite{Bedewy2016,Bedewy2017,BedewyJournal2017,BedewyJournal2017_2,multiflow18} are quite general, which include most  age penalty models considered in previous studies as special cases. For example, because the functional $f(\{\Delta_n: n\geq 0\})$ is a mapping from the space of age processes to  real numbers, it can  be selected to describe the time-average age (i.e., $1/N\sum_{n=0}^N \Delta_n$), or the time-average of an age penalty function that depends on the age levels at multiple time instants (i.e., $1/N\sum_{n=0}^N p(\Delta_n,\Delta_{n-1},\ldots,\Delta_{n-k})$).  

In this paper, we propose a new measure for the freshness of information, which can precisely describe how information ages over time. 
For Markov sources, an online sampling policy is developed to optimize the freshness of information.\footnote{Non-Markov sources will be considered in our future work.} The detailed contributions of this paper are summarized as follows:

%\emph{the mutual information between the real-time source value and the received samples} to measure {the freshness of  information} that is available at the receiver. Compared to the information freshness metrics considered in \cite{Yates2012,SunISIT2017,AgeOfInfo2016,Kosta2017,Bedewy2016,Bedewy2017,BedewyJournal2017,BedewyJournal2017_2,multiflow18},  this mutual information term is easy to compute for a wide classes of information sources and . 

%has several benefits: First, the above mutual information term can   about a time-correlated source. 

%propose to evaluate the freshness of information by using a mutual information term. This method allows us to explore the time-correlation structure 

\begin{itemize}
%\item We %This mutual information term can be used to compare the freshness of information about different sources.

\item We propose to use  the mutual information between the real-time source value and the received  samples to quantify the freshness of the information contained in the received samples. This mutual information term is easy to compute for Markov sources: By using the data processing inequality, this mutual information is shown to be a \emph{non-negative} and \emph{non-increasing} function of the age $\Delta_n$ (Lemma \ref{lem1}).
Therefore, the ``aging'' of the received information can be interpreted as a procedure that this mutual information reduces as the age  $\Delta_n$ grows. 

\item In order to optimize the freshness of information, we study the optimal sampling strategy that  maximizes the time-average expected mutual information. This problem is solved in two steps: \emph{(i)} We first generalize \cite{AgeOfInfo2016} to obtain an optimal sampling strategy that  minimizes the time-average expected age penalty function $\limsup_{N\rightarrow \infty}\frac{1}{N}~\mathbb{E}[\sum_{n=1}^N p(\age_n)]$, where $p(\age)$ is  an \emph{arbitrary non-decreasing}  function of the age $\age$ (Theorem \ref{thm1}). 
%We note that  the  results in \cite{AgeOfInfo2016} only hold for  \emph{positive} and \emph{non-decreasing} age penalty functions.
\emph{(ii)} Next, we apply the result of Step \emph{(i)} to a special age penalty function, i.e., the negative of the mutual information, which is a \emph{non-positive} and \emph{non-decreasing} function of the age. 

\item The obtained optimal sampling strategy has a nice structure: A new sample is taken once a conditional mutual information reduces to a threshold $\beta$, and the threshold $\beta$ is  equal to the optimum value of the time-average expected mutual information that we  are  maximizing (Theorem \ref{thm2}). Numerical results are provided to compare different sampling policies. 

%\item In particular, the conventional uniform sampling policy is far from the optimum because of the queueing delay. 

\end{itemize}

\subsection{Relationship with Previous Work}
The closest study to this paper is \cite{AgeOfInfo2016}. 
The differences between \cite{AgeOfInfo2016} and this paper are explained in   the following:
\begin{itemize}
\item The age penalty function $p(\cdot)$  in \cite{AgeOfInfo2016} is  \emph{non-negative} and \emph{non-decreasing}. It cannot be directly applied to our problem, because the negative of the mutual information is a \emph{non-positive} and \emph{non-decreasing} function of the age. We relaxed $p(\cdot)$ to be an \emph{arbitrary non-decreasing} function in this paper.

\item In \cite{AgeOfInfo2016}, a two-layered nested bisection search algorithm was developed to compute the threshold $\beta$. In this paper, $\beta$ is  characterized  as the solution of a fixed-point equation, which can be solved by a single layer of bisection search. Hence, the computation of $\beta$ is  simplified.

\item In \cite{AgeOfInfo2016}, the optimal sampling strategy was obtained for a continuous-time system. In this paper, we develop an optimal sampling strategy for a discrete-time system, without taking any approximation or sub-optimality. 

\ignore{
\item We note that using Theorem \ref{thm1} and Theorem \ref{thm2}, one can easily obtain the optimal sampling policy of the corresponding continuous-time sampling problems. In particular, let $X_n = U_{n2^{-k}}$ be  samples of a continuous-time Markov chain $U_t$ that are taken periodic with a period $2^{-k}$. By taking the limit $k\rightarrow \infty$ and using the integral convergence theorems in \cite{Durrettbook10}, one can obtain the optimal sampling policy for the continuous-time cases. In the derived continuous-time sampling solutions, the $\min$ operator in \eqref{thm1_eq1} and \eqref{coro2_eq1} will be replaced by $\inf$, and the summations in \eqref{thm1_eq4} %, \eqref{thm1_eq3}, 
and \eqref{coro2_eq2} will become integrals.}

\item It was assume in \cite{AgeOfInfo2016} that after the previous sample was delivered, the next sample must be generated within a fixed amount of time. By adopting more powerful proof techniques, we are able to remove such an assumption and greatly simplify the proof procedure in this paper.

%An upper bound $M>0$ was imposed in \cite{AgeOfInfo2016} such that  $S_{i+1}-D_i\leq M$ holds for all sample $i$. We are able to remove such an assumption in this paper. %{\red undefined}

\end{itemize}

%sThe differences between this paper and \cite{AgeOfInfo2016}  are summarized at the end of Section \ref{sec_main_result}. 

%A comparison with \cite{AgeOfInfo2016} is presented in Section \ref{sec_dis}. 
\section{System Model}

We consider a discrete-time status-update system that is illustrated in Fig. \ref{fig_model}, where samples of a source $X_n$ are taken and sent to a receiver through a communication channel. The channel is modeled as a single-server FIFO queue with  \emph{i.i.d.} service times. The system starts to operate at time instant $n=0$. The $i$-th sample is generated at time instant $S_i$ and is delivered to the receiver at time instant $D_i$ with a discrete service time $Y_i$, where  $S_1\leq S_2\leq\ldots$, $S_i +Y_i \leq D_i$, and $\mathbb{E}[Y_i]<\infty$  for all $i$. Each sample packet contains both the sampling time $S_i$ and the sample value $X_{S_i}$. % are both sent to the receiver.
%Each sample $(S_i, X_{S_i})$ is time-stamped with its generation time $S_i$. 
The samples that the receiver has received by time instant $n$ are denoted by the set
\begin{align}\label{eq_samples}
\bm{W}_n = \{X_{S_i}: D_i \leq n\}. 
\end{align}
At any time instant $n$, the receiver uses the received samples  $\bm{W}_n$
to reconstruct an estimate $\hat X_n$ of the real-time source value $X_n$, where we assume that the estimator neglects the knowledge implied by the timing $S_i$ for taking the samples.

Let $U_n= \max\{S_{i}: D_{i} \leq n\}$ be the time stamp  of the freshest sample that the receiver has received by time instant $n$. Then, the \emph{age of information}, or simply the \emph{age}, at time instant $n$ is defined as  \cite{Song1990,KaulYatesGruteser-Infocom2012}
\begin{align}\label{eq_age}
\Delta_{n} = n-U_n = n - \max\{S_{i}: D_{i} \leq n\}.
\end{align}
The initial state of the system is assumed to satisfy $S_1 = 0$, $D_1 = Y_1$, and $\Delta_{0}$ is a finite constant. 

Let $\pi = (S_1,S_2,\ldots)$ represent a sampling policy %{\red Define the policy space}
 and $\Pi$ denote the set of  \emph{causal} sampling policies that satisfy the following two conditions: \emph{(i)} Each sampling time $S_i$ is chosen based on {history and current information of the system}, but not on any {future information}. \emph{(ii)} The inter-sampling times $\{T_i = S_{i+1}-S_i, i=1,2,\ldots\}$ form a {regenerative process} \cite[Section 6.1]{Haas2002}\footnote{We assume that $T_i$ is a regenerative process because we will optimize $\liminf_{N\rightarrow \infty}\mathbb{E}[\sum_{n=1}^NI(X_n; \bm{W}_n)]/N$, but operationally a nicer objective function is $\liminf_{i\rightarrow \infty}{\mathbb{E}[\sum_{n=0}^{D_i} I(X_n; \bm{W}_n)]}/{\mathbb{E}[D_i]}$. These two objective functions are equivalent if $\{T_1,T_2,\ldots\}$ is a regenerative process.}: There exists an increasing sequence  $0\leq {k_1}<k_2< \ldots$ of almost surely finite random integers such that the post-${k_j}$ process $\{T_{k_j+i}, i=1,2,\ldots\}$ has the same distribution as the post-${k_1}$ process $\{T_{k_1+i}, i=1,2,\ldots\}$ and is independent of the pre-$k_j$ process $\{T_{i}, i=1,2,\ldots, k_j-1\}$; in addition, $0<\mathbb{E}[S_{k_{j+1}}-S_{k_j}]<\infty, ~j=1,2,\ldots$

%Similar to \cite{SOLEYMANI20161}, we assume that the estimator neglects the implied knowledge when no measurement is transmitted.

We assume that the Markov chain $X_n$ and the service times $Y_i$ are determined by two {mutually independent} external processes, which do not change according to the adopted sampling policy.

%does not know how the sampling times $S_i$ are chosen in the sampling policy. 

%there exists $M>0$ such that
%\begin{equation}\label{eq_finite}
%
%\end{equation}

\section{Mutual Information as a Measure of the Freshness of Information}\label{sec:metric}

In this paper, we propose to use the mutual information 
\begin{align}
I(X_n; \bm{W}_n)  = H(X_n) - H(X_n| \bm{W}_n)
\end{align}
as a metric for evaluating the freshness of information that is available at the receiver. In information theory, $I(X_n; \bm{W}_n)$ is the amount of information that the received samples $\bm{W}_n$ carries about the real-time source value $X_n$. If $I(X_n; \bm{W}_n)$ is  close to $H(X_n)$, the received samples $\bm{W}_n$ are considered to be fresh; if $I(X_n; \bm{W}_n)$ is almost  $0$,  the received samples $\bm{W}_n$ are considered to be obsolete. In addition, because $I(X_n; \bm{W}_n)$ has naturally incorporated the information structure of the source $X_n$, it can effectively characterize the freshness of information about  sources with different time-varying patterns.

One way to interpret  $I(X_n; \bm{W}_n)$ is to consider how helpful the received samples $\bm{W}_n$ are for inferring $X_n$. 
By using the Shannon code lengths \cite[Section 5.4]{Cover}, the expected minimum number of bits $L$ required to specify $X_n$ satisfies 
\begin{align}\label{eq_length}
H(X_n) \leq L < H(X_n) + 1,
\end{align}
where $L$ can be interpreted as the expected minimum number of binary tests that are needed to infer $X_n$. 
On the other hand, with the knowledge of $\bm{W}_n$, the expected minimum number of bits $L'$ required to specify $X_n$ satisfies
\begin{align}\label{eq_length1}
H(X_n| \bm{W}_n) \leq L' < H(X_n| \bm{W}_n) + 1.
\end{align}
If $X_n$ is a random vector consisting of a large number of symbols (e.g., $X_n$ represents an image containing many pixels or the channel coefficients of many OFDM subcarriers), the one bit of overhead in \eqref{eq_length} and \eqref{eq_length1} is insignificant. 
Hence, $I(X_n; \bm{W}_n)$ is approximately the reduction in the description cost for inferring $X_n$ without and with the knowledge of $\bm{W}_n$.

%The mutual information between the source value $X_n$ and its estimate $\hat X_n$ is $I(X_n; \hat X_n)$. 
%%%Since the estimator neglects the timing information contained in $S_i$, 
%By data processing inequality \cite{Cover}, $I(X_n; \hat X_n)$ is no greater than $I(X_n; \bm{W}_n)$. For certain estimators  \cite{SunISIT2017}, it holds that $I(X_n; \hat X_n) = I(X_n; \bm{W}_n)$. 

%, which has the following  property:

\subsection{Markov Sources}
To get more insights, let us consider the class of Markov sources and use the Markov property to simplify $I(X_n; \bm{W}_n)$. 
By using the data processing inequality \cite{Cover}, it is not hard to show that  $I(X_n; \bm{W}_n)$ has the following  property:

\begin{lemma}\label{lem1}
If $X_n$ is a time-homogeneous Markov chain and $\bm{W}_n$ is defined in \eqref{eq_samples}, then the mutual information
\begin{align}\label{eq_lem1}
I(X_n; \bm{W}_n) = I(X_n; X_{n-\age_n})
\end{align}
can be expressed as a non-negative and non-increasing function $r(\age_n)$ of the age $\age_n$.
\end{lemma}
\begin{proof}
%\ifreport
Because $X_n$ is a Markov chain, $X_{\max\{S_{i} :  D_{i} \leq n\}}=X_{n-\age_n}$ contains all the information in $\bm{W}_n = \{X_{S_i}: D_i \leq n\}$ about $X_n$. In other words, $X_{n-\age_n}$ is a sufficient statistic of $\bm{W}_n$ for estimating $X_n$. 
Then,  \eqref{eq_lem1} follows from  \cite[Eq. (2.124)]{Cover}.

Next, because $X_n$ is time-homogeneous, $I(X_n; X_{n-\age}) = I(X_{\age+1}; X_{1})$ for all $n$, which is a function of the $\age$. Further, because $X_n$ is a Markov chain, owing to the data processing inequality \cite[Theorem 2.8.1]{Cover}, $I(X_{\age+1}; X_{1})$ is non-increasing in $\age$. Finally,  mutual information  is non-negative. This completes the proof.
\end{proof}

According to Lemma \ref{lem1},  information ``aging''  can be considered as a procedure that \emph{the amount of information $I(X_n; \bm{W}_n)$ that is preserved in $\bm{W}_n$   for {inferring} the real-time source value $X_n$ decreases as the age $\age_n$ grows}. This is similar to the data processing inequality  \cite{Cover} which states that no processing of the data $Y$ can increase the information that $Y$ contains about $Z$; the difference is that in the status-update systems that we consider,  the sample set $\bm{W}_n$, the age $\age_n$, and the signal value $X_n$ are all evolving over time. 
%{\red Explain that timing information cannot increase the freshness mutual information.} 

Two examples of the Markov source $X_n$ are provided in the sequel as illustrations of Lemma \ref{lem1}:
\subsubsection{Gaussian Markov Source} Suppose that $X_n$ is a first-order discrete-time Gaussian Markov process, defined by 
\begin{align}
X_n = a X_{n-1} + Z_n,
\end{align} 
where $a\in(-1,1)$ and the $Z_n$'s are zero-mean \emph{i.i.d.}~Gaussian random variables with variance $\sigma^2$. Because $X_n$ is a Gaussian Markov process, one can show that \cite{Gelfand1959}
\begin{align}
I(X_n; \bm{W}_n) = I\left(X_n; X_{n-\age_n}\right) = -\frac{1}{2} \log_2 \left(1-a^{2\age_n}\right).
\end{align} 
Since $a\in(-1,1)$ and $\age_n\geq 0$ is an integer, $I(X_n; \bm{W}_n)$ is a positive and decreasing function of the age $\age_n$. Note that if $\age_n=0$, then $I(X_n; \bm{W}_n)=H(X_n)=\infty$, because the absolute entropy of a Gaussian random variable is infinite. 
\subsubsection{Binary Markov Source} Suppose that $X_n\in\{0,1\}$ is a binary symmetric Markov chain defined by
\begin{align}\label{eq_binary}
X_n = X_{n-1} \oplus V_n,
\end{align} 
where $\oplus$ denotes binary modulo-2 addition and the $V_n$'s are \emph{i.i.d.} Bernoulli random variables with mean $q\in[0,\frac{1}{2}]$. One can show that
\begin{align}
I(X_n; \bm{W}_n) &= I\left(X_n; X_{n-\age_n}\right) = 1\!-\!h\!\left(\frac{1-(1-2q)^{\age_n} }{2}\right)\!,
\end{align} 
where $\Pr[X_n =1 | X_0=0] = \frac{1-(1-2q)^{n} }{2}$ and $h(x)$ is the binary entropy function defined by $h(x)=-x\log_2 x - (1-x) \log_2 (1-x)$
with a domain $x\in[0,1]$ \cite[Eq. (2.5)]{Cover}. Because $h(x)$ is increasing on $[0,\frac{1}{2}]$, $I(X_n; \bm{W}_n)$ is a non-negative and decreasing function of the age $\age_n$.

\section{Online Sampling for Information Freshness}
In this section, we will develop an optimal online sampling policy that can maximize the freshness of information about Markov sources. 
\subsection{Problem Formulation}

To optimize the freshness of information, we formulate an online sampling problem for maximizing the time-average expected mutual information between $X_n$ and $\bm{W}_n$ over an infinite time-horizon:
\vspace{-0mm}
\begin{equation}\label{eq_problem1}
\bar I_{\text{opt}}=\sup_{\pi\in\Pi}~\liminf_{N\rightarrow \infty}\frac{1}{N}~\mathbb{E}\left[\sum_{n=1}^NI(X_n; \bm{W}_n)\right],
\vspace{-0mm}
\end{equation}
where $\bar I_{\text{opt}}$ is the optimal value of \eqref{eq_problem1}. We assume that $\bar I_{\text{opt}}$ is finite.

It is helpful to remark that $\bar I_{\text{opt}}$ in \eqref{eq_problem1} is different from the Shannon capacity considered in, e.g., \cite{Anantharam1996,Cover}:
In \eqref{eq_problem1}, 
%the received data $\bm{W}_n$ is used to , and 
our goal is to maximize \emph{the freshness of information} and make more accurate inference about the real-time source value; this goal is achieved by minimizing the average amount of mutual information that is lost as the received data becomes obsolete. 
On the other hand, the focus of Shannon capacity theory is mainly on maximizing \emph{the rate of information} that can be reliably transmitted to the receiver, but (in most cases) without significant concerns about whether the received information is new or old. %We will further explore the relationship between the freshness and rate of information in our future work. 

\subsection{Optimal Online Sampling Policy}\label{sec_main_result}

%In this section, we will develop an optimal sampling policy that solves \eqref{eq_problem1}, by exploiting the property of $I(X_n; \bm{W}_n)$ in Lemma \ref{lem1}.
In  \cite{AgeOfInfo2016}, an age penalty function $p(\age)$ was defined to characterize the level of dissatisfaction for having aged information
at the receiver, where $p: \mathbb{R}\mapsto \mathbb{R}$ is an arbitrary \emph{non-negative} and \emph{non-decreasing} function that can be specified according to the  application. For continuous-time status-update systems, the optimal sampling policy for minimizing the time-average expected age penalty $\limsup_{T\rightarrow \infty}\frac{1}{T}~\mathbb{E}[\int_0^T p(\age(t)) dt]$  was obtained in \cite{AgeOfInfo2016}. Unfortunately, we are not able to apply the results in \cite{AgeOfInfo2016} to solve \eqref{eq_problem1}.
%We note that the age penalty model considered in \cite{AgeOfInfo2016} and the information freshness metric proposed in Section \ref{sec:metric} are not compatible: 
Specifically, if we choose an age penalty function $p_2(\age_n) = -I(X_n; \bm{W}_n) = - r(\age_n)$, then Lemma \ref{lem1} suggests that $p_2(\cdot)$ is a \emph{non-positive} and \emph{non-decreasing}, which is different from the \emph{non-negative} and \emph{non-decreasing} age penalty function required in \cite{AgeOfInfo2016}. In addition, we consider a discrete-time system in this paper, which is different from the continuous-time system in \cite{AgeOfInfo2016}. 

\begin{figure}%[!t]
\centering
\includegraphics[width=0.46\textwidth]{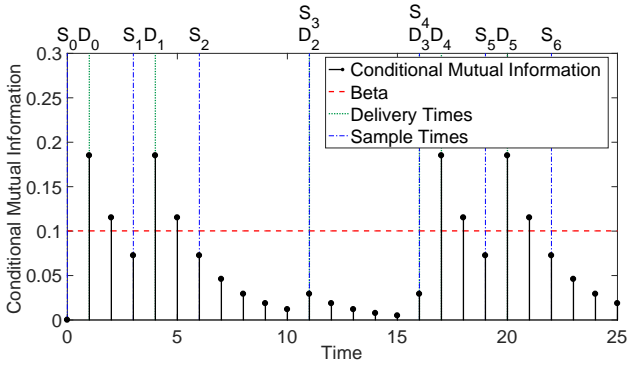}   
\caption{A sample-path illustration of the optimal sampling policy \eqref{coro2_eq1} and \eqref{coro2_eq2}, where the service time $Y_i$ is equal to either $1$ or $5$ with equal probability. 
On this sample-path, the service times are $Y_0=1, Y_1 = 1, Y_2 = 5, Y_3 = 5, Y_4 = 1, Y_5 = 1, Y_6 = 5$.  }
%\vspace{-3mm}
\label{fig_policy}
\end{figure}    
% this optimal sampling policy has a much lower complexity than that in \cite{AgeOfInfo2016}

To address this problem, we  generalize \cite{AgeOfInfo2016} by considering an \emph{arbitrary non-decreasing} age penalty function (no matter \emph{positive} or \emph{negative}) and design an optimal sampling policy that minimizes the time-average expected age penalty. To that end, we consider the following discrete-time age penalty minimization problem: 
\begin{align}
\bar p_{\text{opt}} = \inf_{\pi\in\Pi}~&\limsup_{N\rightarrow \infty}\frac{1}{N}~\mathbb{E}\left[\sum_{n=1}^N p(\age_n)\right]\label{eq_problem} %\\
%\text{s.t.}~~& \liminf_{N\rightarrow \infty} \frac{1}{N}{\mathbb{E}\left[\sum_{n=1}^N (S_{n+1}-S_n)\right]} \geq \frac{1}{f_{\max}},\label{eq_constraint}
\end{align}
where $p: \mathbb{R}\mapsto \mathbb{R}$ is an arbitrary \emph{non-decreasing} function
%, $f_{\max}$ is the maximum allowed sampling rate subject to a time-average resource constraint (i.e., on the power resource or CPU cycles) of the sampler, 
and $\bar p_{\text{opt}}$ denotes the optimal value of  \eqref{eq_problem}. We assume that $\bar p_{\text{opt}}$ is finite.
Problem \eqref{eq_problem} is a Markov decision problem. A closed-form solution of \eqref{eq_problem} is provided in the following theorem:
\ignore{
\begin{theorem}\label{thm1}
If $p: \mathbb{R}\mapsto \mathbb{R}$ is non-decreasing and the service times $Y_i$ are {i.i.d.}, then there exists a threshold $\beta\in\mathbb{R}$ such that the sampling policy 
\begin{align}\label{thm1_eq1}
S_{i+1} \!=\! \min\{n\in \mathbb{N}: n\geq  D_i, \mathbb{E}_{Y_{i+1}}\!\left[p(n\!+\!Y_{i+1}\!-S_i)\right]\! \geq\! \beta \}\!\!
\end{align}
is optimal to \eqref{eq_problem}, where $D_i = S_i + Y_i$, $\mathbb{E}_{Y}$ denotes the expectation with respect to the random variable  $Y$, and $\beta$ is uniquely determined by solving \eqref{thm1_eq1} and \eqref{thm1_eq2}:
\begin{align}\label{thm1_eq2}
\mathbb{E}[D_{i+1}-D_i] = \max \Bigg\{\frac{1}{f_{\max}}, \frac{\mathbb{E}\bigg[\sum\limits_{n=D_i}^{D_{i+1}-1} p(n\!-\!S_i)\bigg]}{\beta}\Bigg\},\!
\end{align}
The optimal value of \eqref{eq_problem} is then given by \emph{
\begin{align}\label{thm1_eq3}
\bar p_{\text{opt}} = \frac{\mathbb{E}\bigg[\sum\limits_{n=D_i}^{D_{i+1}-1} p(n-S_i)\bigg]}{\mathbb{E}[D_{i+1}-D_i]}.
\end{align}}
\end{theorem} }

\begin{theorem}\label{thm1}
If $p: \mathbb{R}\mapsto \mathbb{R}$ in \eqref{eq_problem} is non-decreasing and the service times $Y_i$ are {i.i.d.}, then there exists a threshold $\beta\in\mathbb{R}$ such that the sampling policy 
\begin{align}\label{thm1_eq1}
\!\!\!\!S_{i+1} =\min\{n\in \mathbb{N}\!:\! n\geq  D_i, \mathbb{E}\!\left[p(n+Y_{i+1}-S_i)|S_i,Y_i\right]\! \geq\! \beta \}\!\!
\end{align}
is optimal to \eqref{eq_problem}, where $D_i = S_i + Y_i$ and $\beta$ is determined by solving \eqref{thm1_eq1} and \eqref{thm1_eq4}: 
\begin{align}\label{thm1_eq4}
\beta=\frac{\mathbb{E}\bigg[\sum\limits_{n=D_i}^{D_{i+1}-1} p(n-S_i)\bigg]}{\mathbb{E}[D_{i+1}-D_i]},
\end{align}
%if $\mathbb{E}[S_{i+1}-S_i]   \geq \frac{1}{f_{\max}}$ is satisfied; otherwise, $\beta$ is determined by solving \eqref{thm1_eq1} and \eqref{thm1_eq5}: 
%\begin{align}\label{thm1_eq5}
%\mathbb{E}[S_{i+1}-S_i] = \frac{1}{f_{\max}}.
%\end{align}
Further, $\beta$ is exactly the optimal value of \eqref{eq_problem}, i.e., \emph{$\beta=\bar p_{\text{opt}}$}.
%In addition, the optimal value of \eqref{eq_problem} the  threshold $\beta$ in Theorem \ref{thm1} is exactly the optimal value of  \eqref{eq_problem}, i.e., \emph{$\beta = \bar p_{\text{opt}}$
%In addition, the optimal value of \eqref{eq_problem} is  given by \emph{
%\begin{align}\label{thm1_eq3}
%\bar p_{\text{opt}} = \frac{\mathbb{E}\bigg[\sum\limits_{n=D_i}^{D_{i+1}-1} p(n-S_i)\bigg]}{\mathbb{E}[D_{i+1}-D_i]}.
%\end{align}}
\end{theorem} 

\ifreport
\begin{proof}
See Section \ref{sec_analysis}.
\end{proof}
\else
\begin{proof}[Proof Sketch]
%Due to space limitations, the proof of Theorem \ref{thm1} is relegated to
Theorem \ref{thm1} is proven in three steps: \emph{(i)} We show that it is better to generate samples when the server is idle. This property helps  to remove the harmful queueing delay and simplify the policy space containing the optimal policy. \emph{(ii)} We  use Lemma 2 of \cite{Sun_reportISIT17} to transform the problem into an equivalent problem that is easier to handle. \emph{(iii)} Finally, we show that the optimal sampling problem can be decomposed into a sequence of mutually independent per-sample control problems. The optimal per-sample control policy is given by \eqref{thm1_eq1}, with the threshold $\beta$ satisfying \eqref{thm1_eq4}.
Due to space limitations, the details are relegated to  our technical report \cite{SunSPAWC2018report}. 
\end{proof}
\fi

\ignore{
%Theorem \ref{thm1} tells us that the threshold $\beta$ is determined by \eqref{thm1_eq1} and \eqref{thm1_eq4} if the constraint \eqref{eq_constraint} is inactive; and $\beta$ is determined by \eqref{thm1_eq1} and \eqref{thm1_eq5} if the constraint \eqref{eq_constraint} is active. 
%Theorem \ref{thm1} is more general than Theorem 3 in \cite{AgeOfInfo2016} in several aspects, which will be discussed in Section \ref{sec_dis}. 
%Notice that $\mathbb{E}[S_{i+1}-S_i] = \mathbb{E}[D_{i+1}-D_i]$ is satisfied in Theorem \ref{thm1}.

%From Theorem \ref{thm1}, it is easy to obtain
%
%\begin{corollary}\label{coro1}
%If $f_{\max} =\infty$, then the  threshold $\beta$ in Theorem \ref{thm1} is exactly the optimal value of  \eqref{eq_problem}, i.e., \emph{$\beta = \bar p_{\text{opt}}$}.
%\end{corollary}
In a recent study \cite{TanISIT2018}, an optimal sampling policy was developed for minimizing the age of information for status updates sent from an energy-harvesting source; it was shown that this optimal sampling policy is a threshold policy with multiple thresholds, among which the smallest threshold is equal to the optimum objective value. 
The proof techniques used here are quite different from those in \cite{TanISIT2018}. 
%, the proof techniques used in these two papers are quite different from each other. 
}

Next, we consider a special case that $p(\age_n) = -I(X_n; \bm{W}_n) = - r(\age_n)$. It follows  from Theorem \ref{thm1} %and Corollary \ref{coro1}
 that

\begin{theorem}\label{thm2}
If the service times $Y_i$ are {i.i.d.}, then there exists a threshold $\beta\geq 0$ such that the sampling policy 
\begin{align}\label{coro2_eq1}
\!\!\!\!S_{i+1} &\!=\! \min\{n\!\in\! \mathbb{N}\!: \!n \geq D_i, \nonumber \\
&~~~~~~~~~~~~~\mathbb{E}_{Y_{i+1}}\!\left[I(X_{n+Y_{i+1}}; X_{S_i}| Y_{i+1} = y_{i+1})\right]\! \leq\! \beta \}\!\!\!\!\!\!\nonumber\\
&\!=\! \min\{n\!\in\! \mathbb{N}\!:\! n \geq D_i, I(X_{n+Y_{i+1}}; X_{S_i}| Y_{i+1})\!\leq\! \beta \}\!\!\!
\end{align}
is optimal to \eqref{eq_problem1}, where $D_i = S_i + Y_i$, $\mathbb{E}_{Y}$ denotes the expectation with respect to the random variable  $Y$, and $\beta$ is  determined by solving \eqref{coro2_eq1} and \eqref{coro2_eq2}:
\emph{
\begin{align}\label{coro2_eq2}
\beta=\frac{\mathbb{E}\bigg[\sum\limits_{n=D_i}^{D_{i+1}-1} I(X_n;X_{S_i})\bigg]}{\mathbb{E}[D_{i+1}-D_i]}.
\end{align}}
\!\!Further, $\beta$ is exactly the optimal value of \eqref{eq_problem1}, i.e., \emph{$\beta=\bar I_{\text{opt}}$}.
\end{theorem}

The optimal sampling policy in \eqref{coro2_eq1} and \eqref{coro2_eq2} has a nice structure:
%\noindent\textbf{The Rationale behind the Optimal Sampling Policy:} 
{The next sampling time $S_{i+1}$ is determined based on the mutual information between the freshest received sample ${X}_{S_i}$ and the signal value ${X}_{D_{i+1}}$, where $D_{i+1}=S_{i+1}+Y_{i+1}$ is the delivery time of the $(i+1)$-th sample. ~Because the transmission time $Y_{i+1}$ will be known by both the transmitter and receiver at time $D_{i+1}=S_{i+1}+Y_{i+1}$, $Y_{i+1}$ is the side information that is characterized by the conditional mutual information $I[X_{n+Y_{i+1}}; {X}_{S_i} | Y_{i+1} ]$. The conditional mutual information $I[X_{n+Y_{i+1}}; {X}_{S_i} | Y_{i+1} ]$ decreases as time $n$ grows. According to \eqref{coro2_eq1}, the $(i+1)$-th sample is generated at the smallest integer time instant $n$ satisfying two conditions: \emph{(i)} The $i$-th sample has already been delivered, i.e., $n\geq D_i$, and \emph{(ii)} The conditional mutual information $I[X_{n+Y_{i+1}}; {X}_{S_i} | Y_{i+1} ]$ has reduced to be no greater than a pre-determined threshold $\beta$. In addition, according to \eqref{coro2_eq2}, the threshold $\beta$ is equal to the optimum objective value $\bar I_{\text{opt}}$ in \eqref{eq_problem1}, i.e., the optimum of the time-average expected mutual information $\liminf_{N\rightarrow \infty}\frac{1}{N}~\mathbb{E}[\sum_{n=1}^NI(X_n; \bm{W}_n)]$ that we are maximizing.} Note that the sampling times $S_i$ and delivery times $D_i$ on the right-hand side of \eqref{coro2_eq2} depends on $\beta$. Hence, $\beta$ is a fixed point of \eqref{coro2_eq2}.

 \begin{figure}%[!t]
\centering
\includegraphics[width=0.46\textwidth]{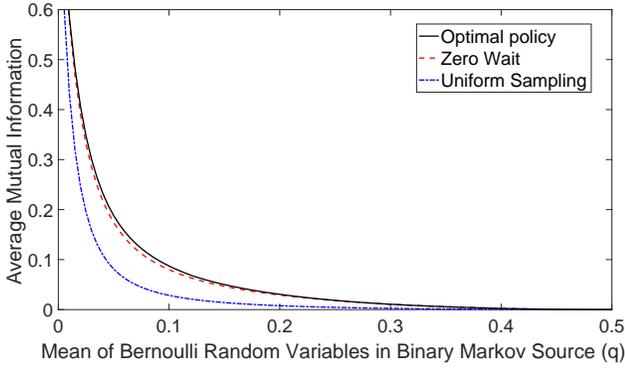}   
\caption{Time-average expected mutual information vs the mean $q$ of Bernoulli random variables $W_n$  for the binary Markov source in \eqref{eq_binary}.}
\label{fig_T2}
%\vspace{-3mm}
\end{figure}

%We note that $\beta$ is \emph{not} the time-average expectation of the conditional mutual information $I[X_{n+Y_{i+1}}; {X}_{S_i} | Y_{i+1} ]$. 
The optimal sampling policy  is illustrated in Fig. \ref{fig_policy}, where the service time $Y_i$ is equal to either $1$ or $5$  with equal probability. The service time $S_i$, delivery time $D_i$, and conditional mutual information $I[X_{n+Y_{i+1}}; {X}_{S_i} | Y_{i+1} ]$ of the samples are depicted in the figure. One can observe that if the service time of the previous sample is $Y_i = 1$, the sampler will wait until the conditional mutual information $I[X_{n+Y_{i+1}}; {X}_{S_i} | Y_{i+1} ]$ drops below the threshold $\beta$ and then take the next sample; if the service time of the previous sample is $Y_i =5$,  the next sample is taken upon the delivery of the previous sample at time $D_i$, because  $I[X_{n+Y_{i+1}}; {X}_{S_i} | Y_{i+1} ]$ is below $\beta$ then.

Notice that in the optimal sampling policy \eqref{coro2_eq1} and \eqref{coro2_eq2}, there is at most one sample in transmission at any time and no sample is waiting in the queue. 
This is different from the traditional uniform sampling policy, in which the waiting time in the queue can be quite high and, as a result, the freshness of information is low. This phenomenon will be illustrated by our numerical results in Section \ref{sec:numerical}.

\ifreport
\section{Proof of Theorem \ref{thm1}}\label{sec_analysis}

%In this section, we prove Theorem \ref{thm1} in two steps: 
%%First, we show that given any sampling policy, the MMSE estimation policy \eqref{eq_esti} is optimal for minimizing the objective function in \eqref{eq_DPExpected}. The remaining task is to solve the optimal sampling problem given an MMSE estimator. Second, 
%First, we  show that no sample should be generated when the server is busy, which simplifies the online sampling problem. Second, we use  \cite[Lemma 2]{Sun_reportISIT17} to further simplify the online sampling problem and then use . the  
%
%in  the Lagrangian dual problem of the simplified problem, and decompose the Lagrangian dual problem into a series of \emph{mutually independent} per-sample control problems. Each of these per-sample control problems has a simple solution. Next, we show that the Lagrangian duality gap of our Markov decision problem is zero. By this, Problem \eqref{eq_problem} is solved. The details are as follows.

\subsection{Simplification of Problem \eqref{eq_problem}}

In \cite{AgeOfInfo2016,SunISIT2017}, it was shown that \emph{no new sample should be taken when the server is busy}. The reason is as follows: If a sample is taken when the server is busy, it has to wait in the queue for its transmission opportunity; meanwhile the sample is becoming stale. A better strategy is to take a new sample once the server becomes idle. By using the sufficient statistic of the Markov chain $X_n$, one can show that the second strategy is better. 

Because of this, we only need to consider a sub-class of sampling policies $\Pi_1\subset\Pi$ in which each sample is generated and submitted to the server after the previous sample is delivered, i.e.,
\begin{align}
\Pi_1 %&= \{\pi\in\Pi: S_{i} = G_{i} \geq D_{j} \text{ for all $i$ and $j=1,\ldots,i$}\} \nonumber \\
 &= \{\pi\in\Pi: S_{i+1} \geq D_{i} = S_i + Y_i \text{ for all $i$}\}. 
\end{align}
Let $Z_i = S_{i+1} -D_{i} \geq0$ represent the waiting time between the delivery time $D_i$ of sample $i$ and the generation time $S_{i+1}$ of sample $i+1$. Since $S_1 = 0$, we have
$S_i =  S_{1} +\sum_{j=1}^{i} (Y_{j} + Z_j)=\sum_{j=1}^{i} (Y_{j} + Z_j)$ and $D_i = S_i + Y_i$. Given $(Y_1,Y_2,\ldots)$, $(S_1,S_2,\ldots)$ is uniquely determined by $(Z_1,Z_2,\ldots)$. Hence, one can also use $\pi = (Z_1,Z_2,\ldots)$ to represent a sampling policy in $\Pi_1$.

Because $T_i $ is a regenerative process, 
using the renewal theory in \cite{Ross1996} and \cite[Section 6.1]{Haas2002}, one can show that in Problem \eqref{eq_problem}, $\frac{1}{i} 
\mathbb{E}[S_i]$ and $\frac{1}{i} 
\mathbb{E}[D_i]$ are convergent sequences and %$(Z_1,Z_2,\ldots)$ is also stationary and ergodic. Hence, 
%Because  $(Y,Y_2,\ldots)$ are \emph{i.i.d.} and $(Z_1,Z_2,\ldots)$ are stationary and ergodic, we can obtain
\begin{align}
&\limsup_{N\rightarrow \infty}\frac{1}{N}~\mathbb{E}\left[\sum_{n=1}^N p(\age_n)\right] \nonumber
\\
=& \lim_{i\rightarrow \infty}\frac{\mathbb{E}\left[\sum_{n=1}^{D_i} p(\age_n)\right]}{\mathbb{E}[D_i]} \nonumber\\
=& \lim_{i\rightarrow \infty}\frac{\sum_{j=1}^{i}\mathbb{E}\left[\sum_{n=D_{j}}^{D_{j+1}-1}p(\age_n)\right]}{\sum_{j=1}^{i} \mathbb{E}\left[Y_j+Z_j\right]}.\nonumber
%= & \lim_{n\rightarrow \infty}\frac{\sum_{i=1}^n\mathbb{E}\left[(W_{S_{i}+Y_i+Z_i}-W_{S_i})^4\right]}{ 6\sum_{i=1}^n \mathbb{E}\left[Y_i+Z_i\right]} +  \mathbb{E}\left[Y\right],\nonumber
\end{align}
In addition, for each policy in $\Pi_1$, it holds that $D_i \leq D_{i+1}$. In this case, the age $\age_n$ in \eqref{eq_age} can be expressed as 
\begin{align}
\age_n = n- S_i,~\text{if}~D_i\leq n < D_{i+1}. \nonumber
\end{align}
Hence,
\begin{align}\label{eq_sum1}
\sum_{n=D_{i}}^{D_{i+1}-1}\!\!p(\age_n)=\!\!\sum_{n=D_{i}}^{D_{i+1}-1}\!\!p(n-S_i)=\!\!\sum_{n=Y_{i}}^{Y_i+Z_i+Y_{i+1}-1}\!\!p(n), 
\end{align}
which is a function of $(Y_i,Z_i,Y_{i+1})$. Define 
\begin{align}\label{eq_sum}
q(Y_i,Z_i,Y_{i+1})= \sum_{n=Y_{i}}^{Y_i+Z_i+Y_{i+1}-1}\!\!p(n), 
\end{align}
then  \eqref{eq_problem} can be simplified as 
\begin{align}\label{eq_problem_S1}
\bar p_{\text{opt}} = \inf_{\pi\in\Pi_1}~&\lim_{i\rightarrow \infty}\frac{\sum_{j=1}^{i}\mathbb{E}\left[q(Y_j,Z_j,Y_{j+1})\right]}{\sum_{j=1}^{i} \mathbb{E}\left[Y_j+Z_j\right]}.
%\\
%\text{s.t.}~~& \lim_{i\rightarrow \infty} \frac{1}{i}\sum_{j=1}^{i} \mathbb{E}\left[Y_j+Z_j\right] \geq \frac{1}{f_{\max}}.
\end{align}

In order to solve \eqref{eq_problem_S1}, let us consider the following Markov decision problem with a parameter $c\geq 0$:
%In order to solve \eqref{eq_Simple}, let us consider the following Markov decision with a parameter $c$:
\begin{align}\label{eq_SD}
\!\!h(c)\!\triangleq\!\inf_{\pi\in\Pi_1}&\lim_{i\rightarrow \infty}\frac{1}{i}\sum_{j=0}^{i-1}\mathbb{E}\left[q(Y_j,Z_j,Y_{j+1})-c(Y_j+Z_j)\right]
\!\!\!\!
%\text{s.t.}~&\lim_{i\rightarrow \infty} \frac{1}{i}\sum_{j=1}^{i} \mathbb{E}\left[Y_j+Z_j\right] \geq \frac{1}{f_{\max}},\nonumber
\end{align}
where $h(c)$ is the optimum  value of \eqref{eq_SD}. Similar with Dinkelbach's method  \cite{Dinkelbach67} for nonlinear fractional programming, the following lemma in \cite{Sun_reportISIT17} also holds for our Markov decision problem \eqref{eq_problem_S1}:

\begin{lemma} \label{lem_ratio_to_minus} \cite[Lemma 2]{Sun_reportISIT17}
The following assertions are true:
\begin{itemize}
\item[(a).] \emph{$\bar p_{\text{opt}}  \gtreqqless c $} if and only if $h(c)\gtreqqless 0$. 
\item[(b).] If $h(c)=0$, the solutions to \eqref{eq_problem_S1}
and \eqref{eq_SD} are identical. 
\end{itemize}
\end{lemma}

Hence, the solution to \eqref{eq_problem_S1} can be obtained by solving \eqref{eq_SD} and seeking $\bar p_{\text{opt}}\in\mathbb{R}$  that satisfies
\begin{align}\label{eq_c}
h(\bar p_{\text{opt}})=0. 
\end{align}

\subsection{Optimal Solution of  \eqref{eq_SD} for \emph{$c = \bar p_{\text{opt}}$} }
%Although \eqref{eq_SD} is not a convex optimization problem, it is possible to use the Lagrangian dual approach to solve \eqref{eq_SD} and show that it admits no duality gap. 
%
%When $c = \bar p_{\text{opt}}$, define the following Lagrangian
%\begin{align}\label{eq_Lagrangian}
%& L(\pi;\lambda) \nonumber\\
%=&\lim_{i\rightarrow \infty}\frac{1}{i}\sum_{j=0}^{i-1}\!\mathbb{E}\!\left[q(Y_j,Z_j,Y_{j+1})\!-\! \bar p_{\text{opt}} (Y_j\!+\!Z_j)\right] + \frac{\lambda}{f_{\max}},
%\end{align}
%where $\lambda\geq0$ is the dual variable. 
%Let
%\begin{align}\label{eq_primal}
%g(\lambda) \triangleq \inf_{\pi\in\Pi_1}  L(\pi;\lambda).
%\end{align}
%Then, the Lagrangian dual problem of \eqref{eq_SD} is defined by
%\begin{align}\label{eq_dual}
%d \triangleq \max_{\lambda\geq 0}g(\lambda),
%\end{align}
%where $d$ is the optimum value of \eqref{eq_dual}. 
%Weak duality \cite{Bertsekas2003,Boyd04} implies that $d \leq h(\bar p_{\text{opt}})$. In Section \ref{sec:dual}, we will establish strong duality, i.e., $d = h(\bar p_{\text{opt}})$.

%\subsection{The Optimal Solution to \eqref{eq_primal}}

%We solve \eqref{eq_primal} in several steps: First, we  prove that there exists a \emph{stationary randomized policy} that is optimal for  \eqref{eq_primal}. Next, we prove that there exists a \emph{stationary deterministic policy} that is optimal for  \eqref{eq_primal}. Finally, we find the optimal stationary deterministic policy. 

Next, we present an optimal solution to \eqref{eq_SD} for $c = \bar p_{\text{opt}}$.  

\begin{definition}
A policy $\pi\in\Pi_{1}$ is said to be a \emph{stationary randomized} policy, if it observes $Y_i$ and then chooses a waiting time $Z_i\in[0,\infty)$ 
based on the observed value of $Y_i$, %. Specifically, $Z_i$ is determined 
according to a conditional probability measure $p(y,A)\triangleq\Pr[Z_i\in A| Y_i=y]$ that is invariant for all $i=1,2,\ldots$ Let $\Pi_{\text{SR}}$ ($\Pi_{\text{SR}}\subset \Pi_1$)  denote the set of stationary randomized policies, defined by
\begin{align}
&\Pi_{\text{SR}}\!=\!\{\pi\in\Pi_1: \text{Given the observation $Y_i=y_i$, $Z_i$ is chosen}\nonumber\\
&\text{according to the probability measure } p(y_i,A) \text{ for all } i\}.\nonumber
\end{align}
\end{definition}

\begin{lemma}\label{lem_SRoptimal} %(Optimality of Stationary Deterministic Policies)
If the service times $Y_i$ are {i.i.d.}, then there exists a stationary randomized policy that is optimal for solving \eqref{eq_SD} with \emph{$c = \bar p_{\text{opt}}$}.
\end{lemma}
\begin{proof}
%The key proof idea is to show that $Y_i$ is a sufficient statistic for determining $Z_i$ in \eqref{eq_primal}.
%See Appendix \ref{app2} for the details.

In \eqref{eq_SD}, the minimization of the term 
\begin{align}\label{eq_opt_stopping1}
&\mathbb{E}\left[q(Y_j,Z_j,Y_{j+1})- \bar p_{\text{opt}} (Y_j+Z_j)\right]\nonumber\\
=&\mathbb{E}\left[q(Y_j,Z_j,Y_{j+1})- \bar p_{\text{opt}} (Z_j+Y_{j+1})\right]
\end{align}
over $Z_j$ depends on $(Y_1,\ldots,Y_j, Z_1,\ldots,Z_{j-1})$ via $Y_j$. Hence, $Y_j$ is a sufficient statistic for determining $Z_j$ in \eqref{eq_SD}. This means that the rule for determining $Z_i$ can be represented by the conditional probability distribution $\Pr[Z_i\in A| Y_i=y_i]$, and in addition, there exists an optimal  solution $(Z_1,Z_2,\ldots)$ to \eqref{eq_SD}, in which $Z_i$ is determined by solving 
\begin{align}\label{eq_opt_stopping2}
\!\!\!\!\min_{\substack{\Pr[Z_i\in A| Y_i=y_i]}} \!\!\!\!\mathbb{E}\left[q(Y_i,Z_i,Y_{i+1})\!-\! \bar p_{\text{opt}} (Z_j+Y_{j+1})\big|Y_i=y_i \right]\!,\!\!
\end{align}
and then use  the observation $Y_i=y_i$ and the optimal conditional probability distribution $\Pr[Z_i\in A| Y_i=y_i]$ that solves \eqref{eq_opt_stopping2}  to decide $Z_i$. Finally, notice that the minimizer of \eqref{eq_opt_stopping2} depends on the joint distribution of $Y_i$ and $Y_{i+1}$. Because the $Y_i$'s are  \emph{i.i.d.},  the joint distribution of $Y_i$ and $Y_{i+1}$ is invariant for $i=1,2,\ldots$ Hence, the optimal conditional probability measure $\Pr[Z_i\in A| Y_i=y_i]$ solving \eqref{eq_opt_stopping2} is invariant for $i=1,2,\ldots$ By definition,  there exists a stationary randomized policy that is optimal for solving Problem \eqref{eq_SD} with $c = \bar p_{\text{opt}}$, which completes the proof.
\end{proof}

%In the sequel, when we refer to the stationary distribution of a stationary randomized policy $\pi\in \Pi_{\text{SR}}$, we will remove  subscript $i$. In particular, {the random variables $(Y_i,Z_i,Y_{i+1})$ are replaced by $(Y,Z,Y')$}, where $Z\geq0$ is chosen based on the conditional probability measure $\Pr[Z\in A| Y=y] =p(y,A)$ after observing $Y=y$, $Y$ and $Y'$ are \emph{i.i.d.} random variables with the same joint distribution as $Y_i$. 
%By Theorem \ref{lem_SRoptimal}, problem \eqref{eq_primal} can be simplified as
%\begin{align}\label{eq_opt_stopping3}
%\min_{p(y,A)} \mathbb{E}\left[q(Y,Z,Y')\!-\! \bar p_{\text{opt}} (Y+Z)\big|Y \right],
%\end{align}
% where $p(y,A) = \Pr[Z\in A| Y=y]$ is the conditional probability measure corresponding to  a stationary randomized policy. 
Next, by using an idea similar to that in the solution of \cite[Problem 5.5.3]{Bertsekas}, we can obtain 
%the solution to \eqref{eq_opt_stopping2} can be obtained as 
% we can obtain
\begin{lemma}\label{lem4}
If $p: \mathbb{R}\mapsto \mathbb{R}$ is non-decreasing and the service times $Y_i$ are {i.i.d.}, then an optimal solution $(Z_1,Z_2,\ldots)$ of \eqref{eq_SD} is given by
%for any given observation $Y=y$, the optimal $Z$ in \eqref{eq_opt_stopping3} is determined by
\begin{align}\label{lem4_eq1}
Z_i = \min\{n\in \mathbb{N}: \mathbb{E}\!\left[p(Y_i+n+Y_{i+1})|Y_i\right]\! \geq\! \beta \},
\end{align}
where \emph{$\beta = \bar p_{\text{opt}}$}.
\end{lemma}
\begin{proof}
%See Appendix \ref{app3}.
Using \eqref{eq_sum} and $\beta = \bar p_{\text{opt}} $, \eqref{eq_opt_stopping2} can be expressed as
\begin{align}\label{eq_opt_stopping4}
\min_{\substack{\Pr[Z_i\in A| Y_i=y_i]}} \mathbb{E}\left[\sum_{n=0}^{Z_i+Y_{i+1}-1}\!\![p(n+Y_i)- \beta] \Bigg| Y_i\right].
\end{align}
It holds that  for $m =1, 2,3,\ldots$
\begin{align}\label{eq_opt_stopping5}
&\mathbb{E}\!\!\left[\sum_{n=0}^{m+Y_{i+1}}\!\!\!\![p(n+Y_i)- \beta] -\!\!\sum_{n=0}^{m+Y_{i+1}-1}\!\!\!\![p(n+Y_i)- \beta] \Bigg| Y_i\right]\nonumber\\
=& \mathbb{E}\!\left[p(Y_i+m+Y_{i+1})- \beta|Y_i\right].
\end{align}
Because $p: \mathbb{R}\mapsto \mathbb{R}$ is non-decreasing, if $Z_i$ is chosen according to  \eqref{lem4_eq1}, we can obtain
\begin{align}
&\mathbb{E}\!\left[p(Y_i+n+Y_{i+1})-\beta |Y_i\right] < 0,~ n = 0, \ldots,  Z_i-1, \\
&\mathbb{E}\!\left[p(Y_i+n+Y_{i+1})-\beta |Y_i\right] \geq 0,~ n \geq Z_i.\label{eq_opt_stopping6}
\end{align}
Based on \eqref{eq_opt_stopping5}-\eqref{eq_opt_stopping6}, it is easy to see that \eqref{lem4_eq1} is the optimal solution to \eqref{eq_opt_stopping4}. 
This completes the proof.
\end{proof}

%\subsection{Zero Duality Gap between \eqref{eq_SD} and \eqref{eq_dual}} \label{sec:dual}
%
%
%Similar to Theorem 7 in \cite{Sun_reportISIT17}, strong duality can be established in the following lemma:
%\begin{lemma}\label{thm_zero_gap}
%If \emph{$c= \bar p_{\text{opt}}$}, the following assertions are true:
%\begin{itemize}
%\vspace{0.5em}
%\item[(a).] The duality gap between \eqref{eq_SD} and \eqref{eq_dual} is zero, i.e., \emph{$d = h(\bar p_{\text{opt}})$}.
%\vspace{0.5em}
%\item[(b).] A common optimal solution to \eqref{eq_problem},  \eqref{eq_problem_S1}, and \eqref{eq_SD} is given by \eqref{thm1_eq4}, \eqref{thm1_eq5}, and \eqref{lem4_eq1}. The optimal value of \eqref{eq_problem} is given by \eqref{thm1_eq3}.
%\end{itemize}
%%Hence, the MMSE-optimal  policy in \eqref{eq_opt_solution} and \eqref{eq_thm1} is an optimal solution to \eqref{eq_SD}. 
%\end{lemma}
%
%\begin{proof}
%The key proof idea is to show that \eqref{eq_SD} has a \emph{geometric multiplier} \cite[Definition 6.1.1]{Bertsekas2003}.\footnote{Note that geometric multiplier is   different from the traditional Lagrangian multiplier.} According to \cite[Proposition 6.2.3]{Bertsekas2003}, if the set of geometric multipliers is non-empty, then the duality gap must be zero, no matter for convex or non-convex optimization problems. 
%The proof details are provided in Appendix \ref{app4}.
%\end{proof}

Hence, Theorem~\ref{thm1} follows from Lemma \ref{lem_ratio_to_minus} and Lemma \ref{lem4}.

\fi

\section{Numerical Results}\label{sec:numerical}

%\begin{figure}%[!t]
%\centering
%\includegraphics[width=0.48\textwidth]{figs/MutualInfo_T1_binary_discrete.eps}   
%\caption{Average mutual information vs $f_{max}$ for \textit{i.i.d.} discrete transmission times}\vspace{-0.0cm}
%\label{fig_T1}
%\end{figure}    

In this section, we evaluate the freshness of information achieved in the following three sampling policies: 
%(or the special case with $p(\Delta_n = -I(X_n;W_n)=-r(\Delta_n))$ in the Theorem 2). then compare it with the constant wait policies and minimum wait policies in Theorem 1, with the constant wait policies and zero wait policies in Theorem 2.
\begin{itemize}
\item \textit{Uniform sampling}: Periodic sampling with a period given by $S_{i+1} - S_i =  \mathbb{E}[Y_i]$.
\item \textit{Zero-wait}: In this sampling policy, a new sample is taken once the previous sample is delivered to the receiver, so that $S_{i+1} = D_i = S_i + Y_i$. %This is feasible when $f_{max} > 1/\mathbb{E}[Y_i].$
\item \textit{Optimal policy}: The sampling  policy given by Theorem~\ref{thm2}. 
\end{itemize}
Let $I_{\text{uniform}}$, $I_{\text{zero-wait}}$, and $I_{\text{opt}}$ be the average mutual information of these three sampling policies.

We consider the binary Markov source $X_n$ in \eqref{eq_binary}. The service time $Y_i$ is equal to either $1$ or $11$ with equal probability.\footnote{The service time distribution is different from that used in Figure \ref{fig_policy}.} 
%are modeled as \emph{i.i.d.} discrete random variable with a probability mass function $\Pr[Y_i=1] = \Pr[Y_i=11]=0.5$. 
%Figure \ref{fig_T1} shows the tradeoff between mutual information and $f_{max}$ with \textit{i.i.d.} transmission times with mean $\mathbb{E}[Y_i] = 6$. The mutual information is calculated with a Bernoulli mean $q = 0.1$. The graph is a stepped function because all sampling times must be a natural number. The optimal policy produces results similar to the uniform sampling policy for small values of $f_{max}$, as the constraint on the optimal policy is active. As $f_{max}$ increases, $I_{\text{uniform}}$ begins to fall. This is because of the queue growth when the sampling period $T$ is close to $\mathbb{E}[Y_i]$. When $f_{max}$ reaches the maximum throughput $1/6 = 0.1\bar{6}$, $I_{\text{uniform}}$ goes to zero because the queue grows infinitely. When $f_{max} > 0.1\bar{6}$, zero-wait is feasible, but the optimal policy is still better, as the constraint becomes inactive and the optimal policy goes to $\bar{I}_{opt}$. 
Figure \ref{fig_T2} depicts the time-average expected mutual information  versus the mean $q$ of the Bernoulli random variables $V_n$ in  \eqref{eq_binary}. One can observe  that $I_{\text{opt}} \ge I_{\text{zero-wait}} \ge I_{\text{uniform}}$ holds for every value of $q$. Notice that because of the queueing delay in the uniform sampling policy, $I_{\text{uniform}}$  is much smaller than $I_{\text{opt}}$ and  $I_{\text{zero-wait}}$. 
In addition, as $q$ grows from 0 to 0.5, the changing speed of the binary Markov source $X_n$ increases and the freshness of information (i.e., the time-average expected mutual information) decreases. When $q = 0.5$, the $X_n$'s form an \emph{i.i.d.} sequence and the freshness of information is zero in all three sampling policies.

\section{Conclusion}
In this paper, we have used mutual information to evaluate the freshness of the received samples that describe the status of a remote source. We have developed an optimal sampling policy that can maximize the time-average expectation of the above mutual information. This optimal sampling policy has been shown to have a nice structure. In addition, we have generalized \cite{AgeOfInfo2016} by finding the optimal sampling strategies for minimizing the time-average expectation of \emph{arbitrary non-decreasing} age penalty functions. % We will consider other related information freshness metrics in our future work. 
%A lemma developed in  \cite{Sun_reportISIT17} has been used to establish these results.

\bibliographystyle{IEEEtran}
\bibliography{ref}
\appendices
% !TEX root = ./heterogeneous_servers.tex
\appendices

\end{document}